\documentclass[letterpaper,USenglish,12pt,notitlepage]{article}

\usepackage{authblk}

\date{December 2016}

\title{Efficient quantum walk on the grid\\ with multiple marked elements}

\author{Peter H{\o}yer}
\author{Mojtaba Komeili}

\affil{\,Department of Computer Science, University of Calgary, Canada\\
  \textup{\texttt{hoyer@ucalgary.ca, mojtaba.komeili@ucalgary.ca}}}

\usepackage{fullpage}
\usepackage{parskip}
\usepackage{theorem}
\newtheorem{theorem}{Theorem}
\newtheorem{lemma}[theorem]{Lemma}
\newtheorem{corollary}[theorem]{Corollary}
\newtheorem{definition}[theorem]{Definition}
\newenvironment{proof}{\begin{trivlist}\item[]{\flushleft\bf Proof }}
   {\qed\end{trivlist}}

\usepackage{url}
\usepackage{hyperref}

\usepackage{array}
\usepackage{amsmath}
\usepackage{amssymb}
\usepackage{latexsym}
\usepackage{amsfonts}
\usepackage{dsfont}

{
\theoremstyle{plain}
\newtheorem{fact}[theorem]{Fact}
}

\def\squareforqed{\hbox{\rlap{$\sqcap$}$\sqcup$}}
\def\qed{\ifmmode\squareforqed\else{\unskip\nobreak\hfil
\penalty50\hskip1em\null\nobreak\hfil\squareforqed
\parfillskip=0pt\finalhyphendemerits=0\endgraf}\fi}
 
\newcommand{\smallspace}{{\mskip 2mu minus 1mu}}

\newcommand{\bigceils}[1]{\raisebox{.1em}{\big\lceil}#1\smallspace\raisebox{.1em}{\big\rceil}}

\newcommand{\eps}{\epsilon}

\newcommand{\name}[1]{\textup{\textsf{#1}}}

\newcommand{\cht}{\name{HT}}
\newcommand{\eht}{\name{HT}^{+}}
\newcommand{\eff}{\textup{eff}}
\newcommand{\effH}{H_\eff}
\newcommand{\effht}{\name{HT}_\eff}
\newcommand{\estH}{\tilde{h}}

\newcommand{\sdis}{\pi}

\newcommand{\ket}[1]{\ensuremath{\vert{#1}\rangle}}
\newcommand{\inner}[2]{\ensuremath{\langle{#1}\vert{#2}\rangle}}

\newcommand{\op}[1]{\ensuremath{\mathsf{#1}}}

\newcommand{\opD}{\op{D}}
\newcommand{\opU}{\op{U}}
\newcommand{\opW}{\op{W}}
\newcommand{\opP}{\op{P}} 
\newcommand{\opS}{\op{S}}
\newcommand{\opC}{\op{C}}
\newcommand{\qq}{\op{E}}

\newcommand{\refl}{\op{Ref}}
\newcommand{\Adj}{\textup{Adj}} 

\newcommand{\tinit}{\textup{init}}
\newcommand{\kinit}{\ket{\tinit}}

\newcommand{\grid}{\textup{grid}}
\newcommand{\torus}{\textup{torus}}

\newcommand{\Pgrid}{\opP_{\grid}}
\newcommand{\Ptorus}{\opP_{\torus}}

\newcommand{\unmarked}{U}
\newcommand{\marked}{{\mathcal M}}

\begin{document}

\maketitle

\begin{abstract}
  We give a quantum algorithm for finding a marked element on the grid when
  there are multiple marked elements.  Our algorithm uses quadratically fewer
  steps than a random walk on the grid, ignoring logarithmic factors.  This is
  the first known quantum walk that finds a marked element in a number of
  steps less than the square-root of the extended hitting time.  We also give
  a new tighter upper bound on the extended hitting time of a marked subset,
  expressed in terms of the hitting times of its members.
\end{abstract}


 \makeatletter
 \def\blfootnote{\gdef\@thefnmark{}\@footnotetext}
 \makeatother
 \blfootnote{To appear in STACS 2017, the 34th International Symposium on
   Theoretical Aspects of Computer Science.}


\vspace*{-4mm}
\enlargethispage{5mm}

\section{Introduction}
\label{section:introduction}

Searching structured and unstructured data is one of the most fundamental
tasks in computer science.  In many search problems in quantum computing, we
are given a set of $N$ elements of which $M$ elements are marked, and our task
is to find and output a marked element.

Search problems have been studied intensively and found many applications,
both classically and quantumly.  The first result on search within quantum
computing was given by Bennett et al.~\cite{BBBV97}, who showed in 1994 that
any quantum algorithm requires $\Omega(\sqrt{N/M})$ steps to find a marked
element.  Grover~\cite{Gro96} showed next that a quantum computer can find
such a marked element in $O(\sqrt{N})$ steps, compared to $\Omega(N)$ for a
classical computer.  This quadratic speed-up was then generalized to arbitrary
unstructured search problems by a generic amplitude amplification
process by Brassard et al.~\cite{BHMT02}.

Grover's algorithm and amplitude amplification are directly applicable to
unstructured global search problems, but not to search problems relying on a
local realization.  Consider we have just inspected one of the $N$ elements
and found that it is not marked, and we want next to inspect another of the
$N$ elements.  Many search problems have the localized property that it is
less costly to inspect an element that is close to the most recently inspected
element, as opposed to inspecting an arbitrary element.  Many probabilistic
algorithms for such problems use random walks, and the quantum analogue of
such are called quantum walks.

\enlargethispage{6mm}
Quantum walks have proven very successful in quantum computing, with
applications in diverse settings such as communication complexity~\cite{AA05},
element distinctness problems~\cite{Amb07,BCJKM13}, testing group
commutativity~\cite{MN07}, and triangle finding~\cite{MSS07,Gall14}.
Excellent surveys on quantum walks, their history and applications,
include~\cite{Amb03,Kem03,San08,Ven12,NRS14}.

The expected number of steps~$H$ required to find a marked element by a random
walk is called the \emph{hitting time}.  The hitting time depends on the
structure being searched as well as the number and locations of the marked
elements.

Quantum walks have been studied for many structures, and in particular for the
torus.  A~torus is a graph containing $N$ vertices laid out in a
two-dimensional square structure.  It is also called a grid or a
two-dimensional lattice.  The first work on the torus was by Aaronson and
Ambainis~\cite{AA05}, who showed that a torus can be searched in $O(\sqrt{N}
\log^2 N)$ steps.  Their breakthrough result is remarkably close to the
quadratic speed-up that is attainable for unstructured search problems, and it
raised the question of determining the limitations of quantum walks in general
and on a torus in particular.

For the torus, Ambainis et al.~\cite{AKR05} next gave a quantum walk using
$O(\sqrt{N} \log N)$ steps.  The question of whether one could find a marked
element any faster was solved Tulsi~\cite{Tul08} who found a quantum algorithm
using $O(\sqrt{N\log N})$ steps, obtained by attaching an ancilla qubit and
thereby modifying the search space.  The above results assume the torus
contains a single marked element.  If there are multiple marked elements, one
can probabilistically reduce the number of marked elements, potentially
incurring an increased cost, and not what we would naturally expect and
desire, a decreased cost.

For general walks, Szegedy~\cite{Sze04} showed in~an influential paper how to
construct a quantum walk from any given symmetric random walk.  Szegedy's
algorithm detects the presence of a marked element in a number of steps of
order $\sqrt{H}$ which is quadratically smaller than classical hitting
time~$H$.  Szegedy's algorithm applies to any number $M$ of marked elements,
but does not necessarily find a marked element.  In some cases, it outputs a
marked element with success probability no better than if we simply sampled
from the stationary distribution.

Magniez et al.~\cite{MNRS11} next showed how phase estimation can be applied
to the larger class of reversible random walk, and gave an algorithm that both
detects and finds a marked element.  Their algorithm applies to any number of
marked elements, but does not guarantee a quadratic speed-up in the hitting
time.  Magniez et al.~\cite{MNRS12} gave a quantum algorithm that detects the
presence of a marked element for any reversible random walk in $O(\sqrt{H})$
steps.  As Szegedy's algorithm, it applies to any number of marked elements,
but it does not necessarily find a marked element.  Magniez et
al.~\cite{MNRS12} also gave a quantum algorithm that finds a unique marked
element in $O(\sqrt{H})$ steps for any state-transitive random walk.

Krovi et al.~\cite{KOR10} next introduced the novel idea of interpolating
walks.  Krovi et al.{} ~\cite{KMOR16} show that interpolated walks can find a
marked element for any reversible random walk, even with multiple marked
elements.  The algorithm does not guarantee a quadratic speed-up when there
are multiple marked elements.  Dohotaru and H{\o}yer~\cite{DH16} introduced
controlled quantum walks and showed that such walks also find a marked element
for any reversible random walk, even with multiple marked elements, but again,
not quadratically faster when there are multiple marked elements.  The quantum
algorithms given in both papers~\cite{KMOR16} and~\cite{DH16} use a number of
steps in the order of a quantity called the \emph{extended hitting time}.

The question of finding a marked element in quadratically fewer steps than by
a random walk when there are multiple marked elements, has thus remained the
main open question.

The torus has continued to be a canonical graph of study.  Ambainis and
Kokainis~\cite{AK15} show that for the torus, the extended hitting time can be
$\Theta(N)$ while the hitting time is $O(1)$ when there are multiple marked
elements.  On the torus, we can find a unique marked element in $O(\sqrt{N
  \log N})$ steps with success probability of order $1/\log N$ by a
continuous-time quantum walk~\cite{CG14} and by a coin-less quantum
walk~\cite{APN15}.  Ambainis et al.~\cite{ABNOR12} show that the algorithm for
the torus in~\cite{AKR05} can be modified, yielding a quantum algorithm that
uses $O(\sqrt{N \log N})$ steps and finds a unique marked element with
constant probability.  Nahimovs and Santos~\cite{NS16} show that the
probability the algorithm of~\cite{AKR05} finds a marked element can be as
small as $O(1/N)$ when there are two marked elements.  Nahimovs and
Rivosh~\cite{NR16} show that the locations of multiple marked elements on the
torus can significantly impact the hitting time.

In this work, we give a quantum algorithm that finds a marked element
quadratically faster than classically, up to a logarithmic factor, on the
torus, no matter the number of marked elements.  This is the first known
quantum algorithm that finds a marked element faster than the square-root of
the extended hitting time.  For some instances, the extended hitting time is a
factor of~$N$ larger than the hitting time.

We also analyze the extended hitting time.  We give a new upper bound on the
extended hitting time and prove that it is convex in the marked subset, with
respect to the stationary distribution.  These results are stated as
Theorem~\ref{theorem:eht_subadditivity} and Corollary~\ref{corollary:ehtbound}
in Section~\ref{section:eht}.  These two results yield in themselves a
simplification of known quantum walks that are based on pruning the number of
marked elements.

We next define and discuss the torus graph in Section~\ref{section:grid}.
A~major obstacle in finding a better quantum algorithm for the torus has been
its locality properties.  In Section~\ref{section:locality}, we investigate
the locality properties of a random walk on the torus, and we turn these into
our advantage, instead of being a disadvantage.  We are sculpturing the
connectivity.  As argued by Meyer and Wong in~\cite{MW15}, connectivity in
itself is a poor indicator of fast quantum search.  The idea of using
properties of the underlying graph to direct the quantum walk to specific
parts of the search space has been used elsewhere, e.g.{} by Le~Gall
in~\cite{Gall14} to obtain the best known quantum algorithm for triangle
finding.

In Section~\ref{section:algorithm}, we give our new quantum algorithm for
finding a marked element on the torus when there are multiple marked elements.
Our algorithm uses quadratically fewer steps than a random walk, ignoring
logarithmic factors.


\section{Bounds on the extended hitting time}
\label{section:eht}

Consider a Markov chain on a discrete finite state space $X$ of size~$N$.  We
represent its transition function as an $N \times N$ matrix~$\opP$.  The
entries of~$\opP$ are real and non-negative.  Entry $P_{yx}$ denotes the
probability of transitioning from state~$x$ to state~$y$ in one step.  The
entries in each column sum to one, implying that $\opP$ is column-stochastic.
We can consider the matrix $\opP$ as the adjacency matrix of an underlying
directed weighted graph.

We assume that the chain $\opP$ is ergodic, which implies that it has a unique
stationary distribution~$\sdis$ satisfying that $\opP \sdis = \sdis$.  It
follows from the Perron--Frobenius theorem that the stationary
distribution~$\sdis$ has real and positive entries.  A~Markov chain is
\emph{ergodic} if its underlying graph is strongly connected and acyclic.

We also assume that $\opP$ is reversible.  A Markov chain is \emph{reversible}
if $\opP_{yx}\sdis_x = \opP_{xy}\sdis_y $ for all states $x, y \in X$ in the
state space.  This condition expresses that the same amount of probability
transition in either direction between any two states~$x$ and~$y$ in the
stationary distribution.  From now on, we will only consider Markov chains
that are both ergodic and reversible, and we will also refer to such chains as
random walks.  Reversibility permits us to apply spectral analysis, following
the seminal work of Szegedy~\cite{Sze04}.

Let $\marked \subset X$ be the subset of marked states, and let $\unmarked = X
\setminus \marked$ be the remaining states which are unmarked.  We form the
\emph{absorbing} walk $\opP'$ from $\opP$ by modifying all outgoing edges from
marked states into self-loops.  That is, if $x \in \marked$ is marked, we set
$\opP'_{xx} = 1$ and $\opP'_{yx} = 0$ for all other states $y \in X \setminus
\{x\}$.  We set $P'_{yx} = P_{yx}$ for all unmarked states $x \in \unmarked$
and all states $y \in X$.  We will interchangeably refer to states as
``elements.''

The main goal of the random walk~$\opP$ is to find a marked state.  The walk
starts in a state drawn from the stationary distribution~$\sdis$.  We keep
applying the transition function until we reach a marked state, at which point
the walk halts.  The \emph{hitting time} is the expected number of steps it
takes for the random walk to find a marked state, and it is denoted by
$\cht(\opP, \marked)$.

We use spectral analysis to study the hitting time of random walks, as in
Szegedy~\cite{Sze04}.  The \emph{discriminant} of any given random walk
$\opP$ is the matrix $\opD(P) = \sqrt{\opP \circ \opP^T\smallspace}$, where
$T$ denotes matrix transposition, and where the Hadamard product $\circ$
denotes entry-wise product and the square-root is taken entry-wise.  The
discriminant is a symmetric real matrix by definition and thus has real
eigenvalues.

We use both the discriminant of the walk $\opP$ and its absorbing
walk~$\opP'$.  The discriminant $\opD(\opP)$ of~$\opP$ has real eigenvalues
$1= \lambda_1 > \lambda_2 \geq \cdots \geq \lambda_N > -1$ with corresponding
eigenvectors $\ket{\lambda_1}, \ket{\lambda_2}, \ldots, \ket{\lambda_N}$.  The
\emph{spectral gap} of $\opP$ is $\delta = 1-\lambda_2$.

The discriminant $\opD(\opP')$ of the absorbing walk $\opP'$ has $|\marked|$
eigenvectors $\ket{x}$ with eigenvalue~$+1$, one for each marked state $x \in
\marked$.  The remaining $N-|\marked|$ eigenvectors $\ket{\lambda'_1},
\ket{\lambda'_2}, \ldots, \ket{\lambda'_{N-|\marked|}}$ have eigenvalues $1 >
\lambda'_1 \geq \lambda'_2 \geq \cdots \geq \lambda'_{N-|\marked|} > -1$ that
are strictly less than one in absolute value and they span the unmarked
subspace.

The hitting time of $\opP$ can then be expressed in terms of the spectra of
the discriminant $\opD(\opP')$ for the absorbing walk.  Let $\ket{\sdis} =
\sum_{x \in X} \sqrt{\sdis_x} \ket{x}$ denote the column vector corresponding
to the stationary distribution, normalized entry-wise.  For any subset $S
\subseteq X$ of elements, let $\eps_S = \sum_{x \in S} \pi_x$ denote the
probability that the stationary distribution is in a state in~$S$.  Let
$\ket{S_\pi} = \frac{1}{\sqrt{\eps_S}}\sum_{x \in S} \sqrt{\pi_x}\ket{x}$
denote the normalized projection of the stationary distribution $\ket{\sdis}$
onto the subspace spanned by elements in a non-empty subset~$S$.  Let
$\ket{S_\pi}$ be the vector of length zero if subset $S$ is empty.  In
particular, we use $\eps_\marked$ and $\eps_\unmarked$, as well as
$\ket{\marked_\pi}$ and $\ket{\unmarked_\pi}$, to denote the quantities for
the marked and unmarked subsets, respectively.

\begin{lemma}[see e.g.~\cite{Sze04,NRS14,KMOR16}]
  \label{lemma:cht}
  The hitting time of a reversible ergodic Markov chain~$\opP$ with marked
  elements~$\marked$ is
  \begin{equation}\label{eq:cht}
    \cht(\opP, \marked) = \sum_{k = 1}^{N-|\marked|} 
    \frac{|\inner{\lambda_k'}{\unmarked_\pi}|^2}{1 - \lambda_k'}.
  \end{equation}
\end{lemma}

The hitting time is the expected number of steps of $\opP'$ required to reach
a marked vertex, starting from a random unmarked vertex, picked according to
the stationary distribution~$\sdis$.  We define the \emph{effective hitting
  time} $\effht(\opP,\marked)$ as the number of steps of $\opP'$ required to
reach a marked vertex with probability at least~$2/3$, again starting from a
random unmarked vertex, picked according to the stationary
distribution~$\sdis$.  By Markov's inequality, the effective hitting time is
at most three times larger than the hitting time.

In his seminal paper, Szegedy~\cite{Sze04} proved that we can \emph{detect}
whether a marked element exists or not quadratically faster by a quantum
algorithm.  If there are $|\marked|>0$ marked elements~$\marked$, it suffices
to run Szegedy's quantum algorithm for
$O\big(\sqrt{\effht(\opP,\marked)}\smallspace\big)$ steps to determine the
existence of a marked element with bounded one-sided
error~\cite{Sze04,MNRS12}.

A~breakthrough for quantum walks with multiple marked elements was achieved by
Krovi et al.~\cite{KOR10,KMOR10,KMOR16}.  They introduced a walk $\opP(s) =
(1-s) \opP + s \opP'$ which is an interpolation between the non-absorbing walk
$\opP$ and the absorbing walk $\opP'$.  The walk is parameterized by a
quantity $0 < s < 1$, which is chosen to be very close to~1, implying that the
walk is almost absorbing.  They prove that their algorithm both detects and
finds a marked element, even when there are multiple marked elements.  The
limitation is that their quantum walk does not necessarily guarantee a
quadratic speedup over a classical walk.  To measure the number of steps of
their algorithm, they introduce a quantity $\eht(\opP,\marked)$ called the
\emph{extended hitting time}.  Their algorithm takes a number of steps that is
of order $\sqrt{\smash{\eht}(\opP,\marked)}$, the square-root of the extended
hitting time, ignoring logarithmic factors.

Their work raises two main questions.  The first question is to determine the
extent to which the extended hitting time can exceed the hitting time.  The
second question is to continue the quest for the discovery of a quantum
algorithm that finds a marked element quadratically faster than a random walk
when there are multiple marked elements.

Ambainis and Kokainis~\cite{AK15} considered the question of determining the
largest possible ratio between the extended hitting time and the hitting time
for a natural search space.  They show that especially for the torus, the
ratio can be exceptionally large by providing an example of a set of marked
elements $\marked$ on the torus for which $\cht(\opP,\marked) \in O(1)$, yet
$\eht(\opP,\marked) \in \Theta(N)$.  That is, the hitting time is a constant,
yet the extended hitting time is linear in the size of the torus.  Searching
with multiple marked elements on the torus in the square-root of the extended
hitting time can be remarkably slow.

In the case there is a single marked element ($\marked = \{m\}$), the extended
hitting time is identical to the hitting time, and thus $\eht(\opP,\{m\}) =
\cht(\opP,\{m\})$.  For multiple marked elements, Ambainis and
Kokainis~\cite{AK15} proved a general upper bound on the extended hitting time
of $\eht(\opP,\marked) \leq \frac{1}{\eps} \frac{1}{\delta}$, which implies
that the extended hitting time can be at most a factor of $\frac{1}{\delta}$
larger than the hitting time $\cht(\opP,\marked)$.  For the torus, the
spectral gap $\delta$ is of order $\frac{1}{N}$, which is so small that it
permits the above ratio of order~$N$ of the extended hitting time over the
hitting time.

To derive an efficient quantum algorithm for the torus for multiple marked
elements, we first provide a new upper bound on the extended hitting time.  We
show that the extended hitting time on a marked set $\marked$ is never more
than the weighted average of the hitting times of any its constituents.

\begin{theorem}\label{theorem:eht_subadditivity}
  Let $\opP$ be a reversible ergodic random walk with stationary distribution
  $\pi$.  Let $\marked = \cup_{i} \marked_i$ be the disjoint union of
  non-empty subsets $\marked_i$ of marked elements.  The extended hitting time
  on~$\marked$ is at most the weighted average of the extended hitting times
  of its subsets,
  \begin{equation*}
    \eht(\opP, \marked) \in 
    O \left( \sum_i \frac{\eps_i}{\eps_\marked} 
      \smallspace \eht(\opP, \marked_i)\right).
  \end{equation*}
  Here $\eps_i = \sum_{x \in \marked_i} \pi_x$ is the probability that $\pi$
  is in marked subset~$\marked_i$, and $\eps_\marked = \sum_{i \in \marked}
  \eps_i$ is the total probability that $\pi$ is in a marked state.
\end{theorem}

As a corollary, by letting each subset $\marked_i$ be a singleton set, we
obtain that the extended hitting time of a set is never more than the
worst-case hitting time of its members.

\begin{corollary}\label{corollary:ehtbound}
  The extended hitting time $\eht(\opP, \marked)$ on a marked subset $\marked$
  is in the order of the maximum of the hitting times $\cht(\opP,\{m\})$ of
  its members $m\in \marked$.
\end{corollary}

There are two technical obstacles in analyzing and understanding the extended
hitting time.  Firstly, it is defined as a limit of the hitting time of the
interpolated walk $\opP(s)$ as the parameter~$s$ approaches~1.  Secondly, it
is expressed in terms of the spectra of the absorbing walk~$\opP'$, and that
spectra changes as we change the set of marked elements.  Fortunately, we can
circumvent both obstacles by applying the following theorem from~\cite{DH16}.

\begin{theorem}[\cite{DH16}]
\label{theorem:eht}
For any reversible ergodic random walk~$\opP$ with marked elements~$\marked$,
\begin{equation}\label{eq:eht}
  \eht(\opP, \marked) \in 
  \Theta\left(\frac{1}{\eps_\marked} \qq(\opP, \marked_\pi)\right).
\end{equation}
\end{theorem}

The theorem expresses the extended hitting time as a product of two factors.
The first factor is $\frac{1}{\eps_\marked}$, the inverse of the probability
that the stationary distribution is in a marked state.  The second factor is
defined below and is a quantity that is expressed in terms of the spectra
of~$\opP$, the original walk and not the absorbing walk~$\opP'$.
Theorem~\ref{theorem:eht} permits us to analyze the extended hitting time by
analyzing the original walk~$\opP$, a task that is often simpler than
analyzing the absorbing walk~$\opP'$.

\begin{definition}\label{def:escape}
  For any normalized vector $\ket{g}$ over the state space of the walk~$\opP$,
  the \emph{escape time} of~$\ket{g}$ is
  \begin{equation}\label{eq:escape}
    \qq(\opP,\ket{g}) = \sum_{k=2}^{N} 
    \frac{\vert \inner{\lambda_k}{g} \vert^2}{1-\lambda_k}.
  \end{equation}
  For any non-trivial subset $S \subseteq X$ of elements, define the escape time
  of $S$ with respect to $\pi$ as $\qq(\opP,S_\pi) = \qq(\opP,\ket{S_\pi})$.
\end{definition}

We will often omit the subscript $\pi$ and simply write $\qq(\opP,S)$ for
$\qq(\opP,S_\pi)$.  By definition, the escape time is a weighted average over
the reciprocals of all of the spectral gaps $1-\lambda_k$ of the original
walk~$\opP$.  It follows the escape time is at most the inverse of the
(smallest) gap $\delta = 1 - \lambda_2$.  Equation~\ref{eq:eht} then permits
us to re-derive that the extended hitting time is at most
$\frac{1}{\eps_\marked}\frac{1}{\delta}$, as shown by Ambainis and
Kokainis~\cite{AK15}.  The escape time is at least~$1/2$ for any normalized
vector $\ket{g}$ orthogonal to the principal eigenvector~$\ket{\lambda_1}$,
since the denominator in Eq.~\ref{eq:escape} is upper bounded by~2.  We next
show that the escape time is at most additive.
 
\begin{lemma}\label{lemma:escape}
  For any two disjoint subsets of elements $S_1$ and $S_2$,
  \begin{equation*}
    \qq(\opP, S_1 \cup S_2) \leq \qq(\opP,S_1) + \qq(\opP,S_2).
  \end{equation*}
\end{lemma}

\begin{proof}
  Fix a distribution over the vertex set~$V$, e.g.{} the stationary
  distribution~$\pi$.  The lemma holds trivially if $S_1$ or $S_2$ is the
  empty set, and thus assume that both sets are non-empty.  Write the
  normalized state $\ket{S_{1} \cup S_{2}} = a \ket{S_{1}} + b \ket{S_{2}}$ as
  a linear combination of the normalized and orthogonal elements $\ket{S_{1}}$
  and $\ket{S_{2}}$.  Noting that $a^2 + b^2 = 1$, by the Cauchy--Schwarz
  inequality, $\vert \inner{\lambda_k}{S_{1} \cup S_{2}}\vert^2$ is then at
  most the sum of $\vert \inner{\lambda_k}{S_{1}}\vert^2$ and $\vert
  \inner{\lambda_k}{S_{2}}\vert^2$.  We can thus upper bound the sum of the
  terms $\frac{\vert \inner{\lambda_k}{S_{1} \cup S_2}\vert^2}{1 - \lambda_k}$
  by one sum over terms of the form $\frac{\vert
    \inner{\lambda_k}{S_{1}}\vert^2}{1 - \lambda_k}$, plus another sum over
  terms of the form $\frac{\vert \inner{\lambda_k}{S_{2}}\vert^2}{1 -
    \lambda_k}$.  The former sum is the escape time of $S_1$, the latter the
  escape time of~$S_2$.
\end{proof}

We next use the sub-additivity of the escape time to prove that the extended
hitting time on a marked subset~$\marked$ is never more than the extended
hitting time of any of constituents.

\begin{proof}[Proof of Theorem~\ref{theorem:eht_subadditivity}]
  Let $\marked = \marked_1 \cup \marked_2$ be a disjoint union of two
  non-empty subsets of marked elements.  Let $\eps_1 = \eps_{\marked_1}$ be the
  probability that the stationary distribution $\sdis$ is in a marked state
  in~$\marked_1$, and let $\eps_2 = \eps_{\marked_2}$ be defined similarly.
  Let $\eps_\marked = \eps_1 + \eps_2$.

  Using Equation~\ref{eq:eht} and Lemma~\ref{lemma:escape}, and omitting
  asymptotic tight factors, write
  \begin{align*}
    \eht(\opP, \marked) 
    &= \frac{1}{\eps_\marked} \qq(\opP,\marked) 
    \leq \frac{1}{\eps_\marked} \big(\qq(\opP,\marked_1) + \qq(\opP,\marked_2)\big) \\
    &= \frac{\eps_1}{\eps_\marked} \frac{1} {\eps_1} \qq(\opP,\marked_1) + 
       \frac{\eps_2}{\eps_\marked} \frac{1} {\eps_2} \qq(\opP,\marked_2) 
    = \frac{\eps_1}{\eps_\marked} \eht(\opP, \marked_1) + 
       \frac{\eps_2}{\eps_\marked} \eht(\opP, \marked_2).
  \end{align*}
  Theorem~\ref{theorem:eht_subadditivity} follows by linearity.
\end{proof}

Corollary~\ref{corollary:ehtbound} has an important and previously
unrecognized consequence.  Consider we are given some computational problem
that has multiple solutions, and assume that we know how to solve the problem
when there is a unique solution.  Then we may be able to device a randomized
polynomial-time reduction that probabilistically makes all solutions but one
into non-solutions, and then find the only remaining solution.  Such pruning
ideas have been used in {e.g.} reducing SAT to unique-SAT~\cite{VV86} and
finding a marked element on the torus by Aaronson and Ambainis~\cite{AA05}.
As~stated in e.g.~\cite{MNRS12}, randomized reductions of multiple marked
elements to a unique marked element may increase the cost by a
poly-logarithmic factor.

\newcommand{\esteps}{\tilde{\eps}}

Theorem~\ref{theorem:eht_subadditivity} yields an alternative to such
reductions.  We simply just run either the controlled quantum walk
of~\cite{DH16} or the interpolated quantum walk of~\cite{KMOR16}.  Both
algorithms take a number of steps in the order of the square-root of the
extended hitting time, ignoring logarithmic factors.  By
Theorem~\ref{theorem:eht_subadditivity}, the extended hitting time of a subset
$\marked$ is upper bounded by the average of the hitting times of its members,
where the average is with respect to the stationary distribution~$\pi$.
Provided we are given an estimate $\esteps$ of $\epsilon_\marked$ satisfying
that $\frac{2}{3} \leq \frac{\esteps}{\epsilon_\marked} \leq \frac{4}{3}$,
then we find a marked element with probability at least~$1/5$.  (Apply
e.g.~Theorem~7 in~\cite{KMOR16} with the value $\eps_2 = \frac{1}{100}$ in
their proof.)

\begin{corollary}
  \label{corollary:findingmultiple}
  Given a reversible ergodic Markov chain $\op{P}$ with marked elements
  $\marked$, and an estimate $\esteps$ of $\epsilon_\marked$ satisfying that
  $\frac{2}{3} \leq \frac{\esteps}{\epsilon_\marked} \leq \frac{4}{3}$, we can
  find a marked state by a quantum walk with probability at least $1/5$ using
  in the order of $\sqrt{\cht(\opP, \{m\})}$ steps, where $m \in \marked$ is
  chosen to maximize the upper bound.
\end{corollary}

One advantage of applying an algorithm that runs in the square-root of the
extended hitting time, is that no direct pruning is necessary.  We do not need
to turn marked elements into non-marked elements.  It suffices to guess an
estimate $\esteps$ of the probability~$\eps_\marked$ of measuring a marked
state in the stationary distribution.  A~second advantage is that the extended
hitting time of a subset can be significantly less than the average of the
hitting times of its members.  The bounds in
Theorem~\ref{theorem:eht_subadditivity} and Lemma~\ref{lemma:escape} are only
upper bounds, not tight bounds.  The bounds can not be improved in general as
they are tight for some instances.  One such example is the case considered by
Nahimovs and Rivosh~\cite{NR16} of the torus with multiple marked elements
packed as densely as possible into a sub-square.


\section{The torus graph}
\label{section:grid}

We consider walks on two-dimensional square torus graphs.  The graph contains
$N = n^2$ vertices organized into $n$ rows and $n$ columns.  There is one
vertex at location $(r,c)$ for each row $0 \leq r< n$ and column $0\leq c<n$.
The graph is directed and every vertex has in-degree~4 and out-degree~4.

We consider two types of boundary conditions.  We define the \emph{torus}
graph in the usual way, with vertices along the boundary connecting to
vertices on the opposite boundary.  A vertex at location $(r,c)$ is connected
to its four neighbors at locations $(r-1,c)$, $(r+1,c)$, $(r,c-1)$, and
$(r,c+1)$, where the addition is modulo~$n$.

We define the \emph{grid} graph to have self-loops on vertices on the
boundary.  A~vertex at location $(r,c)$ has four out-going edges pointing to
the locations $(\max\{r-1,0\},c)$, $(\min\{r+1,n-1\},c)$, $(r,\max\{c-1,0\})$,
and $(r,\min\{c+1,n-1\})$.  Every vertex has in-degree~4 and out-degree~4, as
for the torus.

Prior to this Section, all of our discussions have been for the torus, not the
grid.  Our algorithm given in Section~\ref{section:algorithm} uses both the
torus and grid graphs.  Since one cannot replace a walk on a torus that
crosses the boundary by a walk on a grid without potentially incurring a cost,
we need to clearly distinguish between the two graphs.  Our algorithm in
Section~\ref{section:algorithm} works for both the torus and grid.

We form a random walk $\opP_G$ on a graph~$G$ by taking the adjacency matrix
$\Adj(G)$ of~$G$ and normalize its \emph{columns}.  For every directed edge
$(u,v) \in G$, we set entry $(v,u)$ in $\opP_G$ to be the inverse of the
out-degree of vertex~$u$.  All other entries of~$\opP_G$ are zero.  Since both
the torus and grid are regular graphs $G$ with out-degree~4, their random walk
operators $\opP_G = \frac{1}{4} \Adj(G)$ are scaled versions of their
adjacency matrices.

Our proposed quantum algorithm for finding a marked state on the torus uses
both the torus and grid graphs.  Since the torus and grid are so
closely related, it seems intuitively obvious that walking on either graphs
should have little influence on the complexity of the algorithm.  Indeed, the
escape times on the torus and grid with $N = |V|$ vertices of an element $m
\in V$ are both of order $\log N$.

\begin{fact}The escape times of an element $m \in V$ on the torus and grid are
  both of order $\log N$,
\begin{equation*}
\qq(\Ptorus,\{m\}) \in \Theta(\log N)
\qquad \textup{ and } \qquad 
\qq(\Pgrid,\{m\}) \in \Theta(\log N).
\end{equation*}
\end{fact}

The above fact can be derived from known facts that the hitting time of a
unique element on the torus and grid are of order $N \log N$ and then applying
Theorem~\ref{theorem:eht}.  The fact can also be shown directly by first
computing the spectra for the torus and grid, as done in e.g.~\cite{AKR05} for
the torus, and then applying Definition~\ref{def:escape}.


\section{Locality of random walks on the torus}
\label{section:locality}

To obtain a faster quantum algorithm, we first need to resolve the main
obstacle that a random walk on a torus is localized.

\begin{lemma}\label{lemma:linelocal}
  The probability that a random walk of $T$ steps on an infinite line stays
  within distance $\bigceils{4\sqrt{T}}$ from the initial position is at
  least~$1-1/745$.
\end{lemma}

\begin{proof}
  Consider a walk on a doubly-infinite line. The walk starts in some fixed
  initial position, and we measure distances from this initial position.  The
  probability that the walk is at distance (strictly) more than $k$ from the
  initial position after $\ell$ steps is at most $2 \exp(- \frac{k^2}{2
    \ell})$ by the Azuma--Hoeffding inequality.

  The conditional probability that the walk ends at a distance larger than
  $k$, conditioned on that the walk ever reaches distance $k+1$, is at
  least~$1/2$ by the reflection principle: once the walk reaches distance
  $k+1$, the walk is equally likely to end on either side of that location.

  By Bayes's rule, the probability that the walk reaches distance $k+1$ is
  then at most $4 \exp(- \frac{k^2}{2 \ell})$ which is at most $4 e^{-8}$ when
  $k = \bigceils{4 \sqrt{T}}$.  Finally, $4e^{-8}$ is less than $1/745$.
\end{proof}

Since the grid is the cartesian graph product of two line graphs, we
immediately get that a walk on the grid is also locally contained.

\begin{lemma}\label{lemma:gridlocal}
  The probability that a random walk of $T$ steps on an infinite grid stays
  within distance $\bigceils{4\sqrt{T}}$ in all four directions from the
  initial position is at least~$1-2/745$.
\end{lemma}

Let us say that a walk is \emph{localized} if it stays within distance
$\bigceils{4\sqrt{T}}$ in all four directions from its initial position~$u$.

This locality property implies that we can substitute the global walk by
disjoint local walks.  Consider a torus of size $n \times n$, and fix an
integer $1\leq d < n$.  We cut the torus into $\Theta((\frac{n}{d})^2)$
disjoint graph components by removing edges from the torus.  We remove edges
that cross graph cuts so that each resulting component is a sub-grid (without
self-loops on the boundary vertices) and so that all components have length
and width that are at least $D$ and at most~$D+1$, for some $d \leq D < 2d$.
We next add self-loops to every vertex on the boundaries of the resulting
graph components.  The overall effect is that we have modified the $n \times
n$ torus into $\Theta((\frac{n}{d})^2)$ disjoint grid graphs, each of size
roughly $D \times D$, by turning edges between adjacent components into
self-loops.

Now, consider a random walk of $T$ steps on the torus of size $N=n^2$ starting
from the stationary distribution~$\pi$.  We set $d = 2 \bigceils{4\sqrt{T}}$
and modify the torus into $\Theta(\frac{N}{T})$ disjoint sub-grids as
described above.  We next sample one of these $\Theta(\frac{N}{T})$ sub-grids
according to the stationary distribution $\pi$ for the torus.  That is, we
sample the sub-grid~$G$ with probability $\eps_G = \sum_{v \in G} \pi_v$.  The
next lemma shows that the probability that a random sub-grid~$G$ contains at
least one marked vertex is high.

\begin{lemma}\label{lemma:subgridfraction}
  If a random walk of $T$ steps on the torus of size $N$ finds a marked vertex
  with probability at least~$p$, for $p \geq \frac{1}{74}$, then the
  probability $p_G$ that a random sub-grid, sampled from the
  $\Theta(\frac{N}{T})$ sub-grids as described above, contains at least one
  marked vertex is at least~$\frac{1}{5}p$.
\end{lemma}

\newcommand{\walk}{\omega}

\begin{proof}
  We define the probabilities $p_{ml}$ and $p_{Gl}$ below and prove the
  following three inequalities,
  \begin{equation*}
    p - \frac{2}{745} \leq p_{ml} \leq p_{Gl} \leq 4 p_G.
  \end{equation*}
  By these three inequalities, when $p \geq \frac{1}{74}$ then $p_G \geq
  \frac{1}{5}p$, and the lemma follows.

  Sample a random walk $\walk$ of length~$T$ as follows: Pick a vertex $u$ on
  the torus according to the stationary distribution~$\pi$, and apply the
  stochastic matrix $\Ptorus$ a number of $T$ times, starting at~$u$.

  Let $p_{ml}$ be the probability that a sampled walk is localized and visits
  a marked vertex.  Let $p_{Gl}$ be the probability that a sampled walk is
  localized and visits a vertex that is in a sub-grid that contains a marked
  vertex.  Clearly, $p_{ml} \leq p_{Gl}$, proving the second inequality.

  The unconditional probability that the walk $\walk$ visits a marked vertex
  is at least~$p$.  The unconditional probability that $\walk$ is \emph{not}
  localized is at most $\frac{2}{745}$ by Lemma~\ref{lemma:gridlocal}.  The
  joint probability that $\walk$ visits a marked vertex and is localized, is
  then at least $p-\frac{2}{745}$, proving the first inequality.

  Fix any sub-grid~$G$ and consider the joint probability that the sampled
  walk $\walk$ is localized and visits~$G$.  For a localized walk $\walk$ to
  visit~$G$, its starting vertex must be within distance
  $\bigceils{4\sqrt{T}}$ of~$G$.  Since $G$ has width and length at least
  $2\bigceils{4\sqrt{T}}$, the number of such vertices is at most four times
  the number of vertices in~$G$.  The probability of sampling any of these as
  the starting vertex~$u$ from the stationary distribution~$\pi$, which is
  uniform, is then at most four times the probability of sampling the starting
  vertex~$u$ from~$G$ itself.  This proves the third inequality, and thus also
  the lemma.
\end{proof}

In the next Section, we use this locality property in the design of our
quantum algorithm and give an efficient algorithm for the torus with multiple
marked elements.


\section{An efficient quantum algorithm for the torus}
\label{section:algorithm}

An implementation of a random walk is done as follows.  We first pick a
starting node $v\in V$ for the random walk.  This node is picked according to
the stationary distribution~$\pi$.  The cost of generating~$v$ is called the
\emph{setup cost} and is denoted by~$\opS$.  We next apply the absorbing walk
$\opP'$ for some number $T$ steps.  Each step consists of two parts: We first
check whether the node~$v$ we are currently located at is marked or not.  The
cost of checking whether a node is marked or not is called the \emph{checking
  cost} and is denoted by~$\opC$.  If $v$ is marked, we halt the algorithm and
output~$v$.  If $v$ is not marked, we next apply the operator~$\opP$ once.
The cost of applying $\opP$ is called the \emph{update cost} and is denoted
by~$\opU$.  After repeating these two parts for $T$ steps, we check whether
the final node~$v$ is marked or not.  If so, we output~$v$, and if not, we
output ``failed search.''  A random walk with $T$ steps has total cost $\opS +
(T+1)\opC + T\opU$, which we write as $\opS + T(\opU + \opC)$ by letting the
setup cost include the cost of the first checking cost.  The walk outputs a
marked element with constant probability when $T \in
\Omega(\effht(\Ptorus,\marked))$.

A~quantum walk is implemented similarly~\cite{Sze04,San08,NRS14}.  We first
create some initial state $\kinit$ from the state $\ket{\pi} = \sum_{v \in V}
\sqrt{\pi_v} \ket{v}$, where $\pi$ is the stationary distribution.  We next
apply some number $T_q$ steps of the quantum walk, where each step consists of
one application of each of two quantum operators, one denoted $\refl(\marked)$
and one denoted $\opW(\opP)$, corresponding to the checking and update
operators applied in a random walk.  The algorithm stops after $T_q$ steps in
some final state.  The costs of these three operators are also denoted $\opS$,
$\opC$, and $\opU$, and are in general comparable in cost to the corresponding
operators for the random walk.  The quantum algorithm has total cost $\opS +
T_q(\opU + \opC)$.

A~measurement of the final state of the quantum walk will not necessarily
produce a marked state with constant probability.  Had that been the case,
then we would have had a quantum algorithm that \emph{finds} a marked element
in cost $\opS + T_q(\opU + \opC)$.  Instead, the quantum walk evolves the
initial state $\kinit$ away from the initial state, and this evolution away
from the initial state is sufficient to \emph{detect} that a marked state
exists.  Szegedy~\cite{Sze04} show that after $T_q$ steps, for some $T_q
\in\Theta\big(\sqrt{\effht(\opP,\marked)}\big)$, the final state has overlap
bounded away from~1 with the initial state.  A change by a constant in overlap
can be detected by standard techniques such as the swap test~\cite{BCWW01}.
If the swap test shows the final state is different from the initial state, we
deduce there is a marked state.  We learn that there exists a marked element,
but we do \emph{not} necessarily find one.  It is possible to efficiently
estimate the speed of the change in overlap by applying eigenvalue
estimation~\cite{Kit95} similar to its uses in e.g.{} quantum
counting~\cite{BHT98}, phase estimation~\cite{CEMM98}, and quantum
walks~\cite{MNRS11,MNRS12,KMOR16}.

\begin{theorem}
  \label{theorem:estimateH}
  There exists a quantum algorithm that given a reversible ergodic random walk
  $\opP$ with marked elements~$\marked \subset X$, with probability at least
  $2/3$ performs as follows: (1) it outputs an estimate $\estH \in
  O(\effht(\opP, \marked))$ satisfying that if we apply $\opP'$ for $\estH$
  steps starting from the initial distribution~$\pi$, we find a marked state
  with probability at least~$3/4$ and (2) it has cost in the order of $\opS +
  \sqrt{\estH}(\opU + \opC)$.
\end{theorem}

Theorem~\ref{theorem:estimateH} states that there is a quantum algorithm that,
with probability at least $2/3$, computes an accurate estimate of the
effective hitting time efficiently.  With complementary probability at most
$1/3$, this does not happen.  We can prevent that the algorithm in
Theorem~\ref{theorem:estimateH} never terminates or terminates after a
significant cost.  Let $H_{\textup{unique}} = \effht(\Ptorus, \{m\}) \in
\Theta(N \log N)$ be the effective hitting time for the torus when there is a
unique marked element.  If the algorithm in Theorem~\ref{theorem:estimateH}
has not halted after $\sqrt{H_{\textup{unique}}}$ steps, we halt the algorithm
and output $H_{\textup{unique}}$ as our estimate~$\estH$.

With this, we can give our quantum algorithm for finding a marked element on
the torus.  Our algorithm works for multiple marked elements, finds a marked
element with probability at least $1/\log N$, and has cost in the order of
\begin{equation*}
  \opS + \min\big\{ \sqrt{H\log H},\sqrt{N \log N} \big\}  (\opU + \opC),
\end{equation*}
where $H = \effht(\Ptorus, \marked)$ is the \emph{effective} hitting time on
marked subset~$\marked$.  This is within a poly-logarithmic factor of being a
quadratic speed-up over the cost of a random walk, which has cost the
effective hitting time.

\begin{theorem}[Main]\label{theorem:main}
  There is a quantum algorithm that, given a torus $\Ptorus$ with marked
  vertices $\marked \subset V$, with probability at least $2/3$, outputs a
  marked element $m \in \marked$ with success probability at least
  $\frac{1}{\log N}$ in cost in the order of $\opS + \min\big\{\sqrt{H \log
    H}, \sqrt{N \log N} \big\} (\opU + \opC)$, where $H = \effht(\Ptorus,
  \marked)$ is the effective hitting time.
\end{theorem}

The input to the algorithm in Theorem~\ref{theorem:main} is a torus $\Ptorus$
of size $n \times n$ with some subset $\marked \subset V$ of vertices being
marked.  The algorithm is as follows.

\begin{enumerate}
\item \label{alg:step:estimateH} Compute an estimate $\estH \in
  O(\effht(\Ptorus, \marked))$ using Theorem~\ref{theorem:estimateH}.  If the
  algorithm in Theorem~\ref{theorem:estimateH} has not halted after
  $\sqrt{H_{\textup{unique}}}$ steps, where $H_{\textup{unique}} =
  \effht(\Ptorus, \{m\})$, halt it and use $H_{\textup{unique}}$ as our
  estimate~$\estH$.
\item Set $d = 2 \bigceils{4 \sqrt{\estH}}$.  If $d> n$, then set $d=n$.
\item Divide the $n \times n$ torus into disjoint sub-grids so that each
  sub-grid has length and width between $D$ and $D+1$, for some $d \leq D <
  2d$.
\item Create the initial state $\ket{\pi} = \sum_{v \in \Ptorus} \sqrt{\pi_v}
  \ket{v}$ over all vertices in the torus corresponding to the stationary
  distribution~$\pi$ for the torus.
\item For each vertex $v \in \Ptorus$, assign the name of the sub-grid
  that~$v$ belongs to in an ancilla register, $\sum_{v \in \Ptorus}
  \sqrt{\pi_v} \ket{v} \ket{\textup{subgrid}(v)}$, in superposition.
\item Set $\esteps = \frac{1}{2^k}$, where integer $k$ is picked uniformly
  at random satisfying $1 \leq 2^k < N$.
\item \label{alg:step:subgridwalk} Run a controlled quantum walk on each
  sub-grid for $\Theta(D \sqrt{\log D})$ steps with estimate $\esteps$ in
  superposition over all sub-grids, by conditioning the walk on the name of
  the sub-grid in the ancilla register.
\item Measure the final state, producing a vertex $v$ of the torus.  Check if
  $v$ is marked.  If so, output~$v$.  If not, output ``unsuccessful search.''
\end{enumerate}

We first prove the correctness of the algorithm.  Assume that the first step
of the algorithm outputs a suitable estimate for the effective hitting time as
given in Theorem~\ref{theorem:estimateH}.  This event happens with probability
at least~$\frac{2}{3}$.  Then the probability that a random sub-grid contains
at least one marked vertex is at least~$\frac{1}{5} \times \frac{3}{4} =
\frac{3}{20}$, by Lemma~\ref{lemma:subgridfraction}.  For each of those
sub-grids, by Corollary~\ref{corollary:findingmultiple}, the controlled
quantum walk in step seven finds a marked element with probability at
least~$\frac{1}{5}$, for at least one of the $\log N$ possible values for~$k$.
Note that in step seven, each of the conditional walks on the sub-grids start
the walk on the state corresponding to the stationary distribution for that
sub-grid.  The entire algorithm thus outputs a marked element with probability
at least $\frac{2}{3} \times \frac{3}{20} \times \frac{1}{5} = \frac{1}{50}$
for at least one of the $\log N$ possible values for~$k$.

The cost of the algorithm is easily deduced.  With probability at least~$2/3$,
two properties hold: (1) the estimate $\estH$ computed in the first step is in
the order of the effective hitting time $\effht(\Ptorus,\marked)$, and (2) the
first and seventh steps of the algorithm each uses a number of steps that is in
the order of $\opS + \sqrt{\estH \log \estH} (\opU + \opC)$.  Further, in all
events, each of the seven steps of the algorithm never uses more than in the
order of $\sqrt{N \log N} (\opU + \opC)$ steps of a quantum walk.
Theorem~\ref{theorem:main} follows.

We remark that we can test all $\log N$ possible values for $k$ by testing
each of them in turn.  This will increase the overall cost by a factor of
$\log N$ and lead to an algorithm with constant success probability.
By~applying amplitude amplification~\cite{BHMT02}, the increase in cost can be
improved to being a factor of order~$\sqrt{\log N}$.  By testing each value of
$k$ in increasing order in turn, our algorithm will have the same cost as the
best known quantum algorithms when there is a unique marked element.

We remark that in step seven, we need to run a quantum walk that finds a
marked element in the sub-grid, even if the sub-grid contains multiple marked
elements.  The controlled quantum walk of~\cite{DH16} and the interpolated
walk of~\cite{KMOR16} both do so in $O(D\sqrt{\log D})$ steps, when provided
an estimate of the probability~$\eps_\marked$.  We also remark that one may
omit the conditioning of the quantum walk in step seven by measuring the
ancilla register containing the name of a sub-grid immediately after step
five.  Conducting measurements as early as possible in a quantum algorithm is
frequently used when no further computations are required.  An~early example
of such is the semi-classical quantum Fourier transform by Griffiths and
Niu~\cite{GN96}.


\section{Concluding remarks}
\label{section:concludingremarks}

We have given an efficient quantum algorithm for a finding a marked element on
the torus with probability at least $1/\log N$ when there are multiple marked
elements.  Our algorithm has cost in the order of $\opS + \sqrt{\effH \log
  \effH} (\opU + \opC)$, where $\effH = \effht(\Ptorus,\marked)$ is the number
of steps used by the random walk $\Ptorus$ to find any one of the marked
elements $\marked$.  This is a quadratic speed-up, up to a poly-logarithmic
factor.  It is the first known quantum walk that has cost less than the
square-root of the extended hitting time.  It is, for the torus, an
affirmative answer to the main open question in quantum walks whether it is
possible to find a marked element efficiently when there are multiple marked
elements.

The study of the torus has proven influential for at least two reasons.
Firstly, much progress in quantum walks has been initiated by work on the
cycle and the torus, and then later generalized to arbitrary graphs.
Secondly, the torus is a hard test-case because the ratio between its hitting
time and the reciprocal of its $\eps\delta$ bound~\cite{Sze04} is large.  For
a unique marked element, the hitting time is of order $N \log N$, and the
reciprocal of the $\eps\delta$ bound is of order~$N^2$, and thus almost
quadratically bigger.

In~this work, we have proposed to use localization to our advantage in quantum
search, and not as an obstacle to be overcome.  We show that localization
makes quantum search efficient when there are multiple marked elements on the
torus.

\subparagraph*{Acknowledgements.}  This work was partially supported by NSERC,
the Natural Sciences and Engineering Research Council of Canada and CIFAR, the
Canadian Institute for Advanced Research.



\newcommand{\SortLex}[1]{}


\begin{thebibliography}{10}

\bibitem{AA05}
Scott Aaronson and Andris Ambainis.
\newblock Quantum search of spatial regions.
\newblock {\em Theory of Computing}, 1(4):47--79, 2005.
\newblock \href {http://arxiv.org/abs/quant-ph/0303041}
  {\path{arXiv:quant-ph/0303041}}, \href
  {http://dx.doi.org/10.4086/toc.2005.v001a004}
  {\path{doi:10.4086/toc.2005.v001a004}}.

\bibitem{Amb03}
Andris Ambainis.
\newblock Quantum walks and their algorithmic applications.
\newblock {\em International Journal of Quantum Information}, 1(4):507--518,
  2003.
\newblock \href {http://arxiv.org/abs/quant-ph/0403120}
  {\path{arXiv:quant-ph/0403120}}, \href
  {http://dx.doi.org/10.1142/S0219749903000383}
  {\path{doi:10.1142/S0219749903000383}}.

\bibitem{Amb07}
Andris Ambainis.
\newblock Quantum walk algorithm for element distinctness.
\newblock {\em SIAM Journal on Computing}, 37(1):210--239, 2007.
\newblock \href {http://arxiv.org/abs/quant-ph/0311001}
  {\path{arXiv:quant-ph/0311001}}, \href
  {http://dx.doi.org/10.1137/S0097539705447311}
  {\path{doi:10.1137/S0097539705447311}}.

\bibitem{ABNOR12}
Andris Ambainis, Art{\={u}}rs Ba{\v{c}}kurs, Nikolajs Nahimovs, Raitis Ozols,
  and Alexander Rivosh.
\newblock Search by quantum walks on two-dimensional grid without amplitude
  amplification.
\newblock In {\em Conference on the Theory of Quantum Computation,
  Communication and Cryptography}, TQC'12, pages 87--97, 2012.
\newblock \href {http://arxiv.org/abs/1112.3337} {\path{arXiv:1112.3337}},
  \href {http://dx.doi.org/10.1007/978-3-642-35656-8_7}
  {\path{doi:10.1007/978-3-642-35656-8_7}}.

\bibitem{AKR05}
Andris Ambainis, Julia Kempe, and Alexander Rivosh.
\newblock Coins make quantum walks faster.
\newblock In {\em 16th Annual ACM-SIAM Symposium on Discrete Algorithms},
  SODA'05, pages 1099--1108, 2005.
\newblock \href {http://arxiv.org/abs/quant-ph/0402107}
  {\path{arXiv:quant-ph/0402107}}.

\bibitem{AK15}
Andris Ambainis and M{\=a}rti{\c n}{\v s} Kokainis.
\newblock Analysis of the extended hitting time and its properties, 2015.
\newblock Poster presented at the 18th Annual Conference on Quantum Information
  Processing, QIP'15, Sydney, Australia.

\bibitem{APN15}
Andris Ambainis, Renato Portugal, and Nikolajs Nahimovs.
\newblock Spatial search on grids with minimum memory.
\newblock {\em Quantum Information \& Computation}, 15(13-14):1233--1247, 2015.
\newblock \href {http://arxiv.org/abs/1312.0172} {\path{arXiv:1312.0172}}.

\bibitem{BCJKM13}
Aleksandrs Belovs, Andrew~M. Childs, Stacey Jeffery, Robin Kothari, and
  Fr{\'e}d{\'e}ric Magniez.
\newblock Time-efficient quantum walks for 3-distinctness.
\newblock In {\em 40th International Colloquium on Automata, Languages, and
  Programming}, ICALP'13, pages 105--122, 2013.
\newblock \href {http://arxiv.org/abs/1302.7316} {\path{arXiv:1302.7316}},
  \href {http://dx.doi.org/10.1007/978-3-642-39206-1_10}
  {\path{doi:10.1007/978-3-642-39206-1_10}}.

\bibitem{BBBV97}
Charles~H. Bennett, Ethan Bernstein, Gilles Brassard, and Umesh Vazirani.
\newblock Strengths and weaknesses of quantum computing.
\newblock {\em SIAM Journal on Computing}, 26(5):1510--1523, 1997.
\newblock \href {http://arxiv.org/abs/quant-ph/9701001}
  {\path{arXiv:quant-ph/9701001}}, \href
  {http://dx.doi.org/10.1137/S0097539796300933}
  {\path{doi:10.1137/S0097539796300933}}.

\bibitem{BHMT02}
Gilles Brassard, Peter H{\o}yer, Michele Mosca, and Alain Tapp.
\newblock Quantum amplitude amplification and estimation.
\newblock In Samuel~J. Lomonaco, Jr. and Howard~E. Brandt, editors, {\em
  Quantum Computation and Quantum Information}, volume 305 of {\em AMS
  Contemporary Mathematics}, pages 53--74. American Mathematical Society, 2002.
\newblock \href {http://arxiv.org/abs/quant-ph/0005055}
  {\path{arXiv:quant-ph/0005055}}.

\bibitem{BHT98}
Gilles Brassard, Peter H{\o}yer, and Alain Tapp.
\newblock Quantum counting.
\newblock In {\em 25th International Colloquium on Automata, Languages, and
  Programming}, ICALP'98, pages 820--831, 1998.
\newblock \href {http://arxiv.org/abs/quant-ph/9805082}
  {\path{arXiv:quant-ph/9805082}}, \href {http://dx.doi.org/10.1007/BFb0055105}
  {\path{doi:10.1007/BFb0055105}}.

\bibitem{BCWW01}
Harry Buhrman, Richard Cleve, John Watrous, and Ronald {\SortLex{Wolf}}de~Wolf.
\newblock Quantum fingerprinting.
\newblock {\em Physical Review Letters}, 87(16):167902, 2001.
\newblock \href {http://arxiv.org/abs/quant-ph/0102001}
  {\path{arXiv:quant-ph/0102001}}, \href
  {http://dx.doi.org/10.1103/PhysRevLett.87.167902}
  {\path{doi:10.1103/PhysRevLett.87.167902}}.

\bibitem{CG14}
Andrew~M. Childs and Yimin Ge.
\newblock Spatial search by continuous-time quantum walks on crystal lattices.
\newblock {\em Physical Review A: General Physics}, 89:052337, 2014.
\newblock \href {http://arxiv.org/abs/1403.2676} {\path{arXiv:1403.2676}},
  \href {http://dx.doi.org/10.1103/PhysRevA.89.052337}
  {\path{doi:10.1103/PhysRevA.89.052337}}.

\bibitem{CEMM98}
Richard Cleve, Artur Ekert, Chiara Macchiavello, and Michele Mosca.
\newblock Quantum algorithms revisited.
\newblock {\em Royal Society of London A: Mathematical, Physical and
  Engineering Sciences}, 454(1969):339--354, 1998.
\newblock \href {http://arxiv.org/abs/quant-ph/9708016}
  {\path{arXiv:quant-ph/9708016}}, \href
  {http://dx.doi.org/10.1098/rspa.1998.0164}
  {\path{doi:10.1098/rspa.1998.0164}}.

\bibitem{DH16}
C\u{a}t\u{a}lin Dohotaru and Peter H{\o}yer.
\newblock Controlled quantum amplification, 2016.
\newblock Manuscript. To be presented at the 20th Annual Conference on Quantum
  Information Processing, QIP'17, Seattle, USA.

\bibitem{Gall14}
Fran\c{c}ois {\SortLex{Gall}}Le~Gall.
\newblock Improved quantum algorithm for triangle finding via combinatorial
  arguments.
\newblock In {\em 55th IEEE Symposium on Foundations of Computer Science},
  FOCS'14, pages 216--225, 2014.
\newblock \href {http://arxiv.org/abs/1407.0085} {\path{arXiv:1407.0085}},
  \href {http://dx.doi.org/10.1109/FOCS.2014.31}
  {\path{doi:10.1109/FOCS.2014.31}}.

\bibitem{GN96}
Robert~B. Griffiths and Chi-Sheng Niu.
\newblock Semiclassical {Fourier} transform for quantum computation.
\newblock {\em Physical Review Letters}, 76:3228--3231, 1996.
\newblock \href {http://arxiv.org/abs/quant-ph/9511007}
  {\path{arXiv:quant-ph/9511007}}, \href
  {http://dx.doi.org/10.1103/PhysRevLett.76.3228}
  {\path{doi:10.1103/PhysRevLett.76.3228}}.

\bibitem{Gro96}
Lov~K. Grover.
\newblock A fast quantum mechanical algorithm for database search.
\newblock In {\em 28th Annual ACM Symposium on Theory of Computing}, STOC'96,
  pages 212--219, 1996.
\newblock \href {http://arxiv.org/abs/quant-ph/9605043}
  {\path{arXiv:quant-ph/9605043}}, \href
  {http://dx.doi.org/10.1145/237814.237866} {\path{doi:10.1145/237814.237866}}.

\bibitem{Kem03}
Julia Kempe.
\newblock Quantum random walks: {An} introductory overview.
\newblock {\em Contemporary Physics}, 44(4):307--327, 2003.
\newblock \href {http://arxiv.org/abs/quant-ph/0303081}
  {\path{arXiv:quant-ph/0303081}}, \href
  {http://dx.doi.org/10.1080/00107151031000110776}
  {\path{doi:10.1080/00107151031000110776}}.

\bibitem{Kit95}
Alexei~Yu. Kitaev.
\newblock Quantum measurements and the abelian stabilizer problem, 1995.
\newblock \href {http://arxiv.org/abs/quant-ph/9511026}
  {\path{arXiv:quant-ph/9511026}}.

\bibitem{KMOR10}
Hari Krovi, Fr{\'e}d{\'e}ric Magniez, M{\=a}ris Ozols, and J{\'e}r{\'e}mie
  Roland.
\newblock Finding is as easy as detecting for quantum walks.
\newblock In {\em 37th International Colloquium on Automata, Languages, and
  Programming}, ICALP'10, pages 540--551, 2010.
\newblock \href {http://arxiv.org/abs/1002.2419v1} {\path{arXiv:1002.2419v1}},
  \href {http://dx.doi.org/10.1007/978-3-642-14165-2_46}
  {\path{doi:10.1007/978-3-642-14165-2_46}}.

\bibitem{KMOR16}
Hari Krovi, Fr{\'e}d{\'e}ric Magniez, M{\=a}ris Ozols, and J{\'e}r{\'e}mie
  Roland.
\newblock Quantum walks can find a marked element on any graph.
\newblock {\em Algorithmica}, 74(2):851--907, 2016.
\newblock \href {http://arxiv.org/abs/1002.2419} {\path{arXiv:1002.2419}},
  \href {http://dx.doi.org/10.1007/s00453-015-9979-8}
  {\path{doi:10.1007/s00453-015-9979-8}}.

\bibitem{KOR10}
Hari Krovi, M{\=a}ris Ozols, and J{\'e}r{\'e}mie Roland.
\newblock Adiabatic condition and the quantum hitting time of {M}arkov chains.
\newblock {\em Physical Review A: General Physics}, 82(2):022333, 2010.
\newblock \href {http://arxiv.org/abs/1004.2721} {\path{arXiv:1004.2721}},
  \href {http://dx.doi.org/10.1103/PhysRevA.82.022333}
  {\path{doi:10.1103/PhysRevA.82.022333}}.

\bibitem{MN07}
Fr{\'e}d{\'e}ric Magniez and Ashwin Nayak.
\newblock Quantum complexity of testing group commutativity.
\newblock {\em Algorithmica}, 48(3):221--232, 2007.
\newblock \href {http://arxiv.org/abs/quant-ph/0506265}
  {\path{arXiv:quant-ph/0506265}}, \href
  {http://dx.doi.org/10.1007/s00453-007-0057-8}
  {\path{doi:10.1007/s00453-007-0057-8}}.

\bibitem{MNRS12}
Fr{\'e}d{\'e}ric Magniez, Ashwin Nayak, Peter~C. Richter, and Miklos Santha.
\newblock On the hitting times of quantum versus random walks.
\newblock {\em Algorithmica}, 63(1):91--116, 2012.
\newblock \href {http://arxiv.org/abs/0808.0084} {\path{arXiv:0808.0084}},
  \href {http://dx.doi.org/10.1007/s00453-011-9521-6}
  {\path{doi:10.1007/s00453-011-9521-6}}.

\bibitem{MNRS11}
Fr{\'e}d{\'e}ric Magniez, Ashwin Nayak, J{\'e}r{\'e}mie Roland, and Miklos
  Santha.
\newblock Search via quantum walk.
\newblock {\em SIAM Journal on Computing}, 40(1):142--164, 2011.
\newblock \href {http://arxiv.org/abs/quant-ph/0608026}
  {\path{arXiv:quant-ph/0608026}}, \href {http://dx.doi.org/10.1137/090745854}
  {\path{doi:10.1137/090745854}}.

\bibitem{MSS07}
Fr{\'e}d{\'e}ric Magniez, Miklos Santha, and Mario Szegedy.
\newblock Quantum algorithms for the triangle problem.
\newblock {\em SIAM Journal on Computing}, 37(2):413--424, 2007.
\newblock \href {http://arxiv.org/abs/quant-ph/0310134}
  {\path{arXiv:quant-ph/0310134}}, \href {http://dx.doi.org/10.1137/050643684}
  {\path{doi:10.1137/050643684}}.

\bibitem{MW15}
David~A. Meyer and Thomas~G. Wong.
\newblock Connectivity is a poor indicator of fast quantum search.
\newblock {\em Physical Review Letters}, 114:110503, 2015.
\newblock \href {http://arxiv.org/abs/1409.5876} {\path{arXiv:1409.5876}},
  \href {http://dx.doi.org/10.1103/PhysRevLett.114.110503}
  {\path{doi:10.1103/PhysRevLett.114.110503}}.

\bibitem{NR16}
Nikolajs Nahimovs and Alexander Rivosh.
\newblock Quantum walks on two-dimensional grids with multiple marked
  locations.
\newblock In {\em 42nd International Conference on Current Trends in Theory and
  Practice of Computer Science}, pages 381--391, 2016.
\newblock \href {http://arxiv.org/abs/1507.03788} {\path{arXiv:1507.03788}},
  \href {http://dx.doi.org/10.1007/978-3-662-49192-8_31}
  {\path{doi:10.1007/978-3-662-49192-8_31}}.

\bibitem{NS16}
Nikolajs Nahimovs and Raqueline A.~M. Santos.
\newblock Adjacent vertices can be hard to find by quantum walks, May 2016.
\newblock \href {http://arxiv.org/abs/1605.05598} {\path{arXiv:1605.05598}}.

\bibitem{NRS14}
Ashwin Nayak, Peter~C. Richter, and Mario Szegedy.
\newblock Quantum analogues of {M}arkov chains.
\newblock In {\em Encyclopedia of Algorithms}. Springer Berlin Heidelberg,
  2014.
\newblock \href {http://dx.doi.org/10.1007/978-3-642-27848-8_302-2}
  {\path{doi:10.1007/978-3-642-27848-8_302-2}}.

\bibitem{San08}
Miklos Santha.
\newblock Quantum walk based search algorithms.
\newblock In {\em International Conference on Theory and Applications of Models
  of Computation}, TAMC'08, pages 31--46, 2008.
\newblock \href {http://arxiv.org/abs/0808.0059} {\path{arXiv:0808.0059}},
  \href {http://dx.doi.org/10.1007/978-3-540-79228-4_3}
  {\path{doi:10.1007/978-3-540-79228-4_3}}.

\bibitem{Sze04}
Mario Szegedy.
\newblock Quantum speed-up of {Markov} chain based algorithms.
\newblock In {\em 45th IEEE Symposium on Foundations of Computer Science},
  FOCS'04, pages 32--41, 2004.
\newblock \href {http://dx.doi.org/10.1109/FOCS.2004.53}
  {\path{doi:10.1109/FOCS.2004.53}}.

\bibitem{Tul08}
Avatar Tulsi.
\newblock Faster quantum walk algorithm for the two dimensional spatial search.
\newblock {\em Physical Review A: General Physics}, 78(4):012310, 2008.
\newblock \href {http://arxiv.org/abs/0801.0497} {\path{arXiv:0801.0497}},
  \href {http://dx.doi.org/10.1103/PhysRevA.78.012310}
  {\path{doi:10.1103/PhysRevA.78.012310}}.

\bibitem{VV86}
Leslie~G. Valiant and Vijay~V. Vazirani.
\newblock {NP} is as easy as detecting unique solutions.
\newblock {\em Theoretical Computer Science}, 47:85--93, 1986.
\newblock \href {http://dx.doi.org/10.1016/0304-3975(86)90135-0}
  {\path{doi:10.1016/0304-3975(86)90135-0}}.

\bibitem{Ven12}
Salvador~E. {Venegas-Andraca}.
\newblock Quantum walks: {A}~comprehensive review.
\newblock {\em Quantum Information Processing}, 11(5):1015--1106, 2012.
\newblock \href {http://arxiv.org/abs/1201.4780} {\path{arXiv:1201.4780}},
  \href {http://dx.doi.org/10.1007/s11128-012-0432-5}
  {\path{doi:10.1007/s11128-012-0432-5}}.

\end{thebibliography}
\end{document}